\documentclass[journal,10pt]{IEEEtran}
\usepackage{amsmath,amssymb,amsfonts,amsthm}
\usepackage{graphicx}
\usepackage{times}
\usepackage[numbers]{natbib}
\usepackage{hyperref}
\usepackage{xspace}
\usepackage{breqn}
\usepackage{caption}
\usepackage{subcaption}
\usepackage{bm}
\usepackage{multirow}
\usepackage{hhline}

\def\Xobs{X_{\text{obs}}\xspace} 
\def\Dlb{D_{\text{lb}}\xspace} 

\def\F{\mathcal{F}}  
\def\O{\mathcal{O}}

  \def\B{\mathcal{B}}
 \def\W{\mathcal{W}}

\def\dR{\mathbb{R}}

\def\Reals{\mathbb{R}}

\renewcommand{\leq}{\leqslant}
\renewcommand{\geq}{\geqslant}

\newcommand{\pspace}{{\sc pspace}\xspace}

\newcommand{\opt}{OPT\xspace}

\newtheorem{defin}{Definition}
  \newenvironment{definition}{\begin{defin} \sl}{\end{defin}}
\newtheorem{theo}[defin]{Theorem}
  \newenvironment{thm}{\begin{theo} \sl}{\end{theo}}
\newtheorem{lemma}[defin]{Lemma}
  
\newtheorem{propo}[defin]{Proposition}
  
\newtheorem{coro}[defin]{Corollary}
  \newenvironment{cor}{\begin{coro} \sl}{\end{coro}}
\newtheorem{obse}[defin]{Observation}
  
\newtheorem{conj}[defin]{Conjecture}

\def\complexity{$\tilde{O}\left(m^4+m^2n^2\right)$\xspace}
\def\complexitybold{$\bm{\tilde{O}\left(m^4+m^2n^2\right)}$\xspace}

\newcommand{\Cpp}{C\raise.08ex\hbox{\tt ++}\xspace}

\begin{document}

\title{Motion Planning for Unlabeled Discs with Optimality
  Guarantees\thanks{\authorrefmark{1}K.\ Solovey, O.\ Zamir and D.\
    Halperin are with the Blavatnic School of Computer Science, Tel
    Aviv University, Israel; Email: {\sf
      \{kirilsol,danha\}$@$post.tau.ac.il, orzamir$@$mail.tau.ac.il};
    The work of K.S.\ and D.H.\ has been supported in part by the
    Israel Science Foundation (grant no.\ 1102/11), by the
    German-Israeli Foundation (grant no.\ 1150-82.6/2011), and by the
    Hermann Minkowski-Minerva Center for Geometry at Tel Aviv
    University.}  \thanks{\authorrefmark{2}J.\ Yu is with the Computer
    Science and Artificial Intelligence Lab at the Massachusetts
    Institute of Technology; Email: {\sf jingjin$@$csail.mit.edu}; His
    work has been supported in part by ONR projects N00014-12-1-1000
    and N00014-09-1-1051.}}

\author{\IEEEauthorblockN{Kiril
    Solovey\authorrefmark{1}, Jingjin
    Yu\authorrefmark{2}, Or Zamir\authorrefmark{1} and Dan
    Halperin\authorrefmark{1}}}

\maketitle

\begin{abstract}
  We study the problem of path planning for unlabeled
  (indistinguishable) unit-disc robots in a planar environment
  cluttered with polygonal obstacles.  We introduce an algorithm which
  minimizes the total path length, i.e., the sum of lengths of the
  individual paths.  Our algorithm is guaranteed to find a solution if
  one exists, or report that none exists otherwise. It runs in time
  \complexitybold, where $\bm{m}$ is the number of robots and $\bm{n}$
  is the total complexity of the workspace. Moreover, the total length
  of the returned solution is at most $\bm{\text{\opt}+4m}$, where \opt
  is the optimal   solution cost.  To the best of our knowledge this is
  the first algorithm for the problem that has such guarantees.  The
  algorithm has been implemented in an exact manner and we present
  experimental results that attest to its efficiency.
\end{abstract}

\IEEEpeerreviewmaketitle

\section{Introduction}

\begin{figure*}
\begin{center}
  \includegraphics[width=0.6\textwidth]{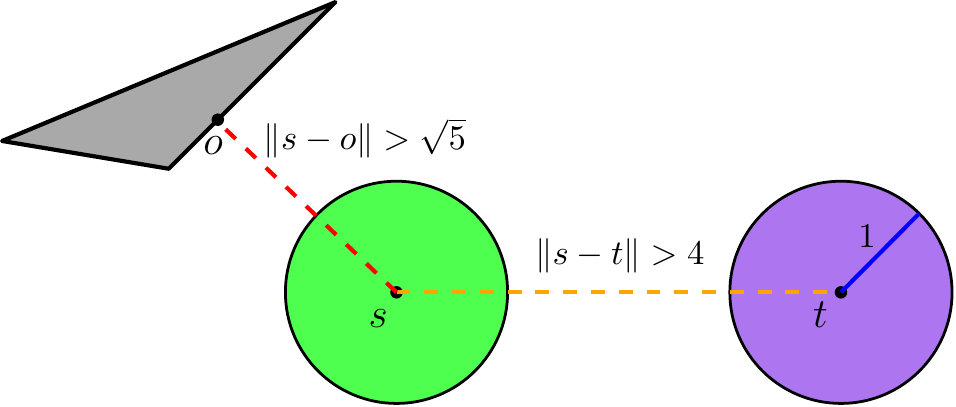}
\end{center}
\caption{Illustration of the separation conditions assumed in this
  work. The green and purple discs represent two unit-disc robots
  placed in $s\in S, t\in T$, respectively. The blue line represents
  the unit radius of the robot (for scale). The distance between $s$
  and $t$ is at least $4$ units (see dashed orange line). The gray
  rectangle represents an obstacle, and the point $o$ represents the
  closest obstacle point to $s$. Notice that the distance between $o$
  and $s$ is at least $\sqrt{5}$ units (see dashed red line).}
  \label{fig:separation}
\end{figure*}

We study the problem of path planning for unlabeled (i.e.,
indistinguishable) unit-disc robots operating in a planar environment
cluttered with polygonal obstacles. The problem consists of moving the
robots from a set of \emph{start} positions to another set of
\emph{goal} positions. Throughout the movement, each robot is required
to avoid collisions with obstacles and well as with its fellow
robots. Since the robots are unlabeled, we only demand that at the end
of the motion each position will be occupied by exactly one robot, but
do not insist on a specific assignment of robots to goals (in contrast
to the standard labeled setting of the problem).

Following the methodology of performing target assignments and path
planning concurrently \cite{calinescu2008reconfigurations,
  YuLav12WAFR, TurMicKum13ICRA} and adopting a graph-based vertex ordering
argument from \cite{YuLav12CDC}, we develop a complete combinatorial
algorithm to the unlabeled problem described above. Our algorithm
makes two simplifying assumptions regarding the input. First, it
assumes that each start or goal position has a (Euclidean) distance of
at least $\sqrt{5}$ to any obstacle in the environment (recall that
the robot discs have unit radius) . Additionally, it requires that the
distance between each start or goal position to any other such
position will be at least $4$ (see Figure~\ref{fig:separation}). Given
that the two separation conditions are fulfilled, out algorithm is
guaranteed to generate a solution if one exists, or report that none
exists otherwise. It has a running time\footnote{For simplicity of presentation, we omit $\log$
  factors when stating the complexity of the algorithm, and hence use
  the $\tilde{O}$ notation.} of \complexity, where $m$ is the number
of robots, and $n$ is the description complexity of the workspace
environment, i.e., the number of edges of the polygons.
Furthermore, the total length of the returned solution, i.e., the sum
of lengths of the individual paths, is at most $\text{\opt}+4m$, where
\opt is the optimal solution cost.

This paper should be juxtaposed with another work by the
authors~\cite{SolHal15} in which we show that a slightly different
setting of the unlabeled problem is computationally intractable. In
particular, we show that the unlabeled problem of unit-square robots
translating amid polygonal obstacles is \pspace-hard. Our proof relies
on a construction of gadgets in which start and goal positions are
very close to one another or to the obstacles, i.e., no separation is
assumed. This can be viewed as a justification to the assumptions used
in the current paper which allow an efficient solution of the problem.

We make several novel contributions. To the best of our knowledge, we
are the first to describe a complete algorithm that fully addresses
the problem of planning minimum total-distance paths for unlabeled
discs. We mention that Turpin et al.~\cite{TurMicKum13ICRA} describe a
complete algorithm for unlabeled discs; however their algorithm
minimizes the maximal path length of an individual robot. Furthermore,
our algorithm makes more natural assumptions on the input problem than
the ones made by Turpin et al. We also mention the work by Adler et
al.~\cite{abhs-unlabeled14} in which a complete algorithm is
described, but it does not guarantee optimality. While the latter work
makes the same assumption as we make on a separation of $4$ between
starts and goals, it does not assume separation from
obstacles. However, it requires that the workspace environment
will consist of a simple polygon. On the practical side, we provide an
exact and efficient implementation which does not rely on
non-deterministic procedures, unlike sampling-based algorithms. The
implementation will be made publicly available upon the publication of
this paper.

Our work is motivated by theoretical curiosity as well as potential
real world applications. From a theoretical standpoint, we have seen
great renewed interests in developing planning algorithms for
multi-robot systems in continuous domains (see, e.g.,
\cite{ TurMicKum13ICRA,TurMicKum12, KarGerSta12, klb-ftpp13,
  TurMohMicKum13}) and discrete settings (see, e.g.,
\cite{YuLav12WAFR, YuLav13ICRA-A}). Whereas significant headway has been
made in solving multi-robot motion-planning problems, many
challenges persist; the problem studied in this paper---path planning
for unlabeled disc robots in a general environment with guarantees on
total distance optimality---remains unresolved (until now). On the
application side, in the past few years, we have witnessed the rapid
development and adaptation of autonomous multi-robot and multi-vehicle
systems in a wide variety of application domains. The most prominent
example is arguably the success of Kiva Systems, now part of Amazon,
which developed a warehousing system employing hundreds of mobile
robots to facilitate order assembly tasks \cite{WurDanMou08}. More
recently, Amazon, DHL, and Google have demonstrated working prototypes
of aerial vehicles capable of automated package delivery. Since the
vehicles are intended to operate in an autonomous, swarm-like setting,
we can foresee in the near future the emerging demand of efficient
path planning algorithms designed for such systems. We note that when
the warehouse robots or the delivery aerial vehicles do not carry
loads, they are effectively unlabeled robots. In such scenarios,
planning collision-free, total-distance near-optimal paths translates into
allowing the vehicles, as a whole, to travel with minimum energy
consumption.

The rest of the paper is organized as follows. In
Section~\ref{sec:related} we review related work. In
Section~\ref{sec:preliminaries} we provide an overview of our
algorithm and necessary background material. In Section~\ref{sec:theory} we
establish several basic properties of the problem and describe the
algorithm in Section~\ref{sec:algorithm}. We report on experimental
results in Section~\ref{sec:experiments} and conclude with a
discussion in Section~\ref{sec:discussion}.

\section{Related work}\label{sec:related}
The problem of multi-robot motion planning is notoriously challenging
as it often involves many degrees of freedom, and consequently a vast
search space~\cite{hss-cmpmio, sy-snp84}. In general, each additional
robot introduces several more degrees of freedom to the
problem. Nevertheless, there is a rich body of work dedicated to this
problem. The earliest research efforts can be traced back to the 1980s
\cite{schwartz1983piano}.

Approaches for solving the problem can be typically subdivided into
categories. Decoupled techniques (see, e.g., \cite{ErdLoz86,
  LavHut98b, PenAke02, BerOve05, GhrOkaLav05, BerSnoLinMan09}) reduce
the size of the search space by partitioning the problem into several
subproblems, which are solved separately, and then the different
components are combined. In contrast to that, centralized
approaches (see, e.g., \cite{OdoLoz89, SveOve98,avbsv-mpfmr, KloHut06,
  shh12, ssh-fne13, WagCho15}) usually work in the combined
high-dimensional configuration space, and thus tend to be slower
than decoupled techniques. However, centralized algorithms often come
with stronger theoretical guarantees, such as completeness. Besides
these, the multi-robot motion-planning problem has also been attacked
using methods based on network flow \cite{KarGerSta12} and mixed integer
programming \cite{GriAke05}, among others.

Multi-robot motion planning can also be considered as a discrete
problem on a graph~\cite{kms-cpmg}. In this case the robots are
pebbles placed on the vertices of the graph and are bound to move from
one set of vertices to another along edges. Many aspects of the
discrete case are well understood. In particular, for the labeled
setting of the problem there exist efficient feasibility-test
algorithms~\cite{ampp-ltafpm, gh-mcpm, Yu13}, as well as complete
planners (\cite{klb-ftpp13, k-cpmg, LunBek11}). For the unlabeled
case, there even exist complete and efficient planners that generate
the optimal solution~\cite{YuLav12CDC, YuLav13ICRA-A, KatYuLav13}
under different metrics. While there exists a fundamental difference
between the discrete and the continuous setting of the multi-robot
problem, the continuous case being exceedingly more difficult, several
recent techniques in the continuous domain~\cite{TurMicKum13ICRA,
  abhs-unlabeled14, sh-kcolor14} have employed concepts that were
initially introduced in the discrete domain. 

As mentioned above, the works by Adler et al.~\cite{abhs-unlabeled14}
and Turpin et al.~\cite{TurMicKum13ICRA} solve very similar settings
of the unlabeled problem for disc robots to the one treated in this
paper, only with different assumptions and goal functions, and thus it
is important to elaborate on these two techniques.  Adler et
al.~\cite{abhs-unlabeled14} show that the unlabeled problem in the
continuous domain can be transformed into a discrete
\emph{pebble-motion on graph} problem. Their construction guarantees
that in case a solution to the former exists, it can be generated by
solving a discrete pebble problem and adapted to the continuous
domain. In particular, a motion of a pebble along an edge is
transformed into a motion of a robot along a local path in the free
space. Turpin et al.~\cite{TurMicKum13ICRA} find an assignment between
starts and goals which minimize the longest path length traveled by
any of the robots. Given such an assignment a shortest path between
every start position to its assigned goal is generated. However, such
paths still may result in collisions between the robots. The authors
show that collisions can be elegantly avoided by prioritizing the
paths. Our current work follows to some extent a similar approach,
although in our case some of the robots must slightly stray from the
precomputed paths in order to guarantee completeness---a thing which
makes our technique much more involved. This follows from the fact
that we make milder assumptions on the input. Another difference is
that the goal function of our algorithm is to minimize the total path
length, which requires very different machinery than the one used by
Turpin et al.

\section{Preliminaries and algorithm overview}\label{sec:preliminaries}
Our problem consists of moving $m$ indistinguishable unit-disc robots
in a workspace~$\W\subset \Reals^2$ cluttered with polygonal
obstacles, whose overall number of edges is $n$.  We define
$\O := \Reals^2\setminus \W$ to be the complement of the workspace,
and we call $\O$ the \emph{obstacle space}.  For given
$r \in \Reals_+, x\in \Reals^2$, we define $\B_r(x)$ to be the open
disc of radius $r$, centered at $x$. For given
$r\in \Reals_+,X\subset \dR^2$ we also define
$\B_r(X):=\bigcup_{x\in X}\B_r(x)$.

We consider the unit-disc robots to be open sets. Thus a robot avoids
collision with the obstacle space if and only if its center is at
distance at least~1 from $\O$. More formally, the collection of all
collision free configurations, termed the \emph{free space}, can be
expressed as
$\F := \left\{x \in \Reals^2 : \B_1(x) \cap \O = \emptyset \right\}$.
We also require the robots to avoid collisions with each other. Thus
for every given two configurations $x,x'\in \F$ two distinct robots
can be placed in $x$ and $x'$ only if $\|x-x'\|\geq 2$. Throughout
this paper the notation $\|\cdot\|$ will indicate the $L_2$ norm.

In addition to the workspace $\W$ we are given sets
$S = \{s_1, s_2, ..., s_m\}$ and $T = \{t_1, t_2, ..., t_m\}$ such
that $S, T \subset \F$. These are respectively the sets of
\emph{start} and \emph{goal} configurations of our $m$
indistinguishable disc robots.  The problem is now to plan a
collision-free motion for $m$ unit-disc robots such that each of them
starts at a configuration in $S$ and ends at a configuration in $T$.
Since the robots are indistinguishable or \emph{unlabeled}, it does
not matter which robot ends up at which goal configuration, as longs
as all the goal configurations are occupied at the end of the
motion. Formally, we wish to find $m$ paths $\pi_i:[0,1]\rightarrow \F$,
for $1\leq i\leq m$, such that $\pi_i(0)=s_i$ and
$\bigcup_{i=1}^m\pi_i(1)=T$. Furthermore, we are interested in finding
a set of such paths which minimizes the expression
$\sum_{i=1}^m|\pi_i|$, where $|\pi_i|$ represents the length of
$\pi_i$ in the $L_2$ norm.

\subsection{Simplifying assumptions}\label{sec:simply}
By making the following simplifying assumptions (see
Figure~\ref{fig:separation}) we are able to show that our algorithm is
complete and near-optimal. The first assumption that we make is
identical to the one used by Adler et al.~\cite{abhs-unlabeled14}. It
requires that every pair of start or goal positions will be separated
by a distance of at least $4$:
\begin{align}\label{equation:separation}
    \forall v, v' \in S\cup T, \quad  \| v - v' \| \geq 4.
\end{align}
\noindent
The motivation for the above assumption is the ability to prove the
existence of a \emph{standalone goal}---a goal position that does not
block other paths, assuming that the paths minimize the total length
of the solution.

We also need the following assumption in order to guarantee that the
robots will be able switch targets, in case that a given assignment of
robots to goals induces collision between the robots:
\begin{align}\label{equation:obstacle-clearance}
  \forall v \in S\cup T \text{\,\,\,and\,\,\,} \forall x\in \O,\quad \|v - x\|
  \geq \sqrt{5}.
\end{align}

\subsection{Overview of algorithm}
Here we provide an overview of our technique. Recall that our problem
consists of $S,T$, which specify the set of start and goal positions
for a collection of $m$ unit-disc robots, and a workspace environment
$\W$. 

We describe the first iteration of our recursive algorithm.  For every
$s_i\in S,t_j\in T$ we find the shortest path
$\gamma_i^j:[0,1]\rightarrow \F$ from $s_i$ to $t_j$. Among these
$m^2$ paths we select a set of $m$ paths
$\Gamma=\{\gamma_1,\ldots, \gamma_m\}$, where
$\gamma_i:[0,1]\rightarrow\F$ for every $1\leq i\leq m$, such that
$\bigcup_{i=1}^m\{\gamma_i(0)\}=S, \bigcup_{i=1}^m\{\gamma_i(1)\} =
T$.
Furthermore, we require that $\Gamma$ will be the path set that
minimizes the total length of its paths under these conditions.  Note
that at this point we only require that the robots will not collide
with obstacles, and do not worry about collisions between the
robots. The generation of $\Gamma$ is described in detail in
Section~\ref{sec:algorithm}, Theorem~\ref{thm:complexity}.

In the next step we find a goal $t\in T$ which does not block paths in
$\Gamma$ that do not lead to $t$. We call such a $t$ a
\emph{standalone goal}. Next, we find a start $s\in S$ from which a
robot will be able to move to $t$ without colliding with other robots
situated in the rest of the start positions. We carefully select such
a start $s$ and generate the respective path to $t$ in order to
minimize the cost of the returned solution. We prepare the input to
the next iteration of the algorithm by assigning
$S:=S\setminus \{s\}, T:=T\setminus \{t\}$, and by treating the robot
placed in $t$ as an obstacle, i.e., $\W:=\W\setminus \B_1(t)$.

\section{Theoretical foundations}\label{sec:theory}
In this section we establish several basic properties of the
problem. Recall that our problem is defined for a workspace~$\W$,
whose free space for a single unit-disc robot is $\F$. Additionally,
we have two sets of start and goal positions
$S=\{s_1,\ldots,s_m\}, T=\{t_1,\ldots,t_m\}$, respectively.

The following lemma implies that if a robot moving from some start
position $s\in S$ along a given path collides with a region $\B_1(t)$,
for some other goal $t\in T$, then there exists another path
$\gamma \in \F$ which moves the robot from $s$ to $t$. In the
  context of our algorithm this lemma implies that if a path for a
  robot interferes with some other goal position, then the path can be
  modified such that it will move the robot to the interfering goal
  instead.

\begin{lemma}\label{lem:shortcuts}
  Let $v\in S\cup T$ and $x\in \F$ such that
  $\B_1(x)\cap \B_1(v)\neq \emptyset$. Then the straight-line path
  from $x$ to $v$ is contained in $\F$, i.e., $\overline{xv}\subset \F$.
\end{lemma}
\begin{proof}
  Here we use the fact that for every $o\in \O$ it holds that
  $\|v-o\|\geq \sqrt{5}$, which is the second separation
  assumption. Without loss of generality, assume that the
  straight-line segment from $v$ to $x$ is parallel to the
  $x$-axis. Denote by $A$ and $B$ the bottom points of the unit discs
  around $v$ and $x$, respectively (see
  Figure~\ref{fig:shortcuts}). Similarly, denote by $C$ and $D$ the
  top points of these discs. By definition of $v,x$, we know that
  $\|v-x\|<2$. Thus, $\|v-B\|=\|v-D\|< \sqrt{5}$ (see dashed red
  segment). This implies that the rectangle defined by the points
  $A,B,C,D$ is entirely contained in $\W$ (see orange square). Thus,
  the straight-line path $\overline{xv}$ is fully contained in $\F$.
\end{proof}

\begin{figure}[b]
\begin{center}
  \includegraphics[width=0.3\textwidth]{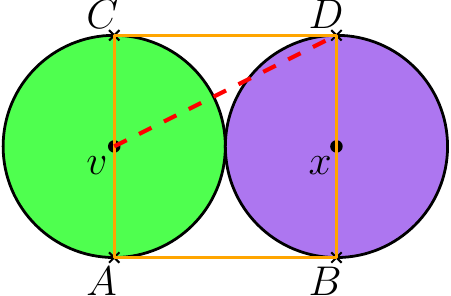}
\end{center}
  \caption{Illustration of the proof of Lemma~\ref{lem:shortcuts}.}
  \label{fig:shortcuts}
\end{figure}

\begin{definition}
  Let $\Gamma$ be a set of $m$ paths $\{\gamma_1,\ldots, \gamma_m\}$
  such that for every $1\leq i\leq m$, $\gamma_i:[0,1]\rightarrow \F$,
  $S=\bigcup_{i=1}^m\{\gamma_i(0)\},T=\bigcup_{i=1}^m\{\gamma_i(1)\}$. We
  call $\Gamma$ the \emph{optimal-assignment path set} for $S,T,\F$,
  if it minimizes the expression $|\Gamma|:=\sum_{i=1}^m|\gamma_i|$,
  over all such path sets.
\end{definition}

Note that $\Gamma$ is not necessarily a feasible solution to our
problem since at this stage we still ignore possible collisions
between robots.  Let $\Gamma=\{\gamma_1,\ldots,\gamma_m\}$ be an
optimal-assignment path set for $S,T,\F$. Without loss of generality,
assume that for every $1\leq i\leq m$,
$\gamma_i(0)=s_i,\gamma_i(1)=t_i$.

\begin{definition}
  Given an optimal-assignment path set $\Gamma$, we call $t_k\in T$ a
  \emph{standalone goal}, for some $1\leq k\leq m$, if for every
  $\gamma_i\in \Gamma, i\neq k$ it holds that
  $\B_1(t_k)\cap \B_1(\gamma_i)=\emptyset$.
\end{definition}

Standalone goals play a crucial role in our algorithm. We first show
that at least one such goal always exists.

\begin{thm}\label{thm:standalone}
  Let $\Gamma$ be an optimal-assignment path set. Then there exists a
  standalone goal.
\end{thm}

\begin{proof}
  Assume by contradiction that every goal $t_i$ interferes with some
  path $\gamma_j\in \Gamma, i\neq j$. This implies that there is a
  circular interference, i.e., there exist $\ell\leq m$ goals, which we
  denote, for simplicity and without loss of generality, as
  $t_{1},\ldots, t_{\ell}$ such that for every $1<i\leq \ell$,
  $\B_1(t_i)\cap \B_1(\gamma_{i-1})\neq \emptyset$, and
  $\B_1(t_1)\cap \B_1(\gamma_{\ell})\neq \emptyset$. More formally,
  let $\mathcal{I}$ be a directed graph vertices are $T$. For every
  $t_i$ which interferes with a path $\gamma_j$, where $j\neq i$ we
  draw an edge from $t_i$ to $t_j$. If there exists no standalone goal
  then $\mathcal{I}$ has a directed cycle of size greater than
  one. This is due to the fact that if $\mathcal{I}$ were
  \emph{directed acyclic} then it should have had a node whose
  \emph{out degree} is zero.

  We show that in this case, the paths
  $\gamma_1, \ldots, \gamma_{\ell}$ do not induce the minimal
  assignment for the starts and goals
  $\{s_1,\ldots, s_{\ell}\},\{t_1,\ldots,t_{\ell}\}$. This would imply
  that $\Gamma$ is not the optimal assignment path set. We claim that
  instead of assigning $s_i$ to $t_i$ we may assign $s_i$ to $t_{i+1}$
  for $1 \leq i <\ell$ and from $s_{\ell}$ to $t_1$ and get a
  collection of paths $\gamma'_1,\ldots, \gamma'_{\ell}$, such that
  $|\gamma'_i|<|\gamma_i|$ for every $1\leq i\leq \ell$. Denote by
  $x\in \gamma_i$ the first interference point with $t_{i+1}$ along
  $\gamma_i$. Additionally, denote by $\gamma_i^x$ the subpath of
  $\gamma_i$ that starts with $s_i$ and ends with $x$. Define $\gamma'_i$
  to be the concatenation of $\gamma_i^x$ and $xt_{i+1}$.

  We need to show first that $\gamma'_i\subset\F$ and that
  $|\gamma'_i|<|\gamma_i|$. The subpath $\gamma^x_i$ is obstacle-collision free,
  and so is $\overline{xt_{i+1}}$ according to Lemma~\ref{lem:shortcuts}. Thus,
  $\gamma'_i$ is free as well. Now, note that $\|x-t_{i+1}\|<2$ and
  $\|t_{i+1}-t_{i}\|\geq 4$. Thus, by the triangle inequality
  $\|x-t_{i}\|\geq 2$. This finishes our proof.
\end{proof}  

We introduce the notion of \emph{$0$-hop} and \emph{$1$-hop}
paths. Informally, a $0$-hop path is a path assigned to the standalone
goal which is not blocked by any start position.

\begin{figure}[b]
  \begin{center}
    \includegraphics[width=0.4\textwidth]{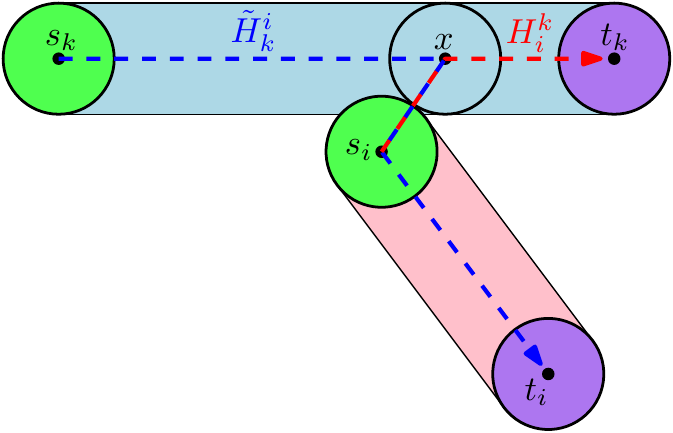}
  \end{center}
  \caption{Example of $1$-hop and switch paths. $t_k$ is
    a standalone goal. The $1$-hop path $H_i^k$ is drawn as a dashed
    red curve, while the switch path $\tilde{H}_k^i$ is drawn as a
    dashed blue curve.}
  \label{fig:paths}
\end{figure}

\begin{definition}
  Let $t_k$ be a standalone goal. The path $\gamma_k$ is called a
  $0$-hop path if for every $s_i\in S,i\neq k$ it follows that $\B_1(s_i)\cap \B_1(\gamma_k)=\emptyset$.
\end{definition}

\begin{definition}\label{def:1hop}
  Let $t_k$ be a standalone goal, and suppose that $\gamma_k$ is not
  $0$-hop. Let $x\in \gamma_k$ be the farthest point along $\gamma_k$
  for which there exists $s_i\in S$ such that
  $\B_1(s_i)\cap \B_1(x) \neq \emptyset$. The $1$-hop path, which is
  denoted by $H_i^k$, is a concatenation of the straight-line path
  $\overline{s_ix}$ and a subpath of $\gamma_k$ that starts at $x$
  and ends~in~$t_k$.
\end{definition}

Namelly, the $1$-hop path moves the robot situated in $s_i$, which
interferes with $\gamma_k$, to $t_k$ (see Figure~\ref{fig:paths}). The
following theorem shows that a robot situated in $s_i$ that blocks the
$0$-hop path can be moved to $t_k$ without inducing collisions with
other robots.

\begin{thm}\label{thm:1_hop}
  Let $t_k$ be a standalone goal. Suppose that $\gamma_k$ is not
  $0$-hop and let $H_i^k$ be the $1$-hop path from $s_i$ to
  $t_k$. Then it holds that for every $s_j\in S, j\neq i,k$,
  $\B_1(H_i^k)\cap \B_1(s_j)=\emptyset$.
\end{thm}

\begin{proof}
  By separation of start positions, and by definition of $x$ (see
  Definition~\ref{def:1hop}), there exists a single start position
  $s_i$ for which $\B_1(x)\cap \B_1(s_i)\neq \emptyset$. By
  Lemma~\ref{lem:shortcuts}, and the separation condition, the path
  $\overline{s_ix}$, does not induce collisions between a robot moving
  along it and other robots placed in $s_j$, for
  $1\leq j\leq m, j\neq i,k$. Finally, by definition of $x$, the rest
  of the path $H_i^k$, which is a subpath of $\gamma_k$, is free of
  collisions with a robot situated~in~$s_j$.
\end{proof}

After moving a robot from $s_i$ to $t_k$ we have to ensure that some
other robot can reach $t_i$. We define the following \emph{switch}
path (see Figure~\ref{fig:paths}).

\begin{definition}
  Let $t_k$ be a standalone goal and suppose that $\gamma_k$ is not
  $0$-hop. The \emph{switch} path $\tilde{H}_k^i$ is a concatenation
  of the following paths: (1)~the subpath of $\gamma_k$ that starts in
  $s_k$ and ends in $x$; (2)~$\overline{xs_i}$; (3)~$\gamma_i$.
\end{definition}

\begin{lemma}\label{lem:switch}
  Let $\tilde{H}_i^k$ be a switch path. Then
  $\B_1(\tilde{H}_i^k)\cap \B_1(t_k)=\emptyset$.
\end{lemma}

\begin{proof}
  $\tilde{H}_i^k$ can potentially interfere with $t_k$ only along
  $\overline{xs_i}$, due to the definition of $x$. For any $y\in
  \overline{xs_i}$ it follows that $\|y-s_i\|<2$. If $\|y-t_k\|<2$
  then it follows that $\|s_i-t_k\|<4$, which is a contradiction.
\end{proof}

\begin{cor}\label{cor:cost_1_hop}
  Let $H_i^k,\tilde{H}_k^i$ be the $1$-hop and switch paths,
  respectively. Then
  $|H_i^k|+|\tilde{H}_k^i|=|\gamma_i|+|\gamma_k|+4$.
\end{cor}

\section{Near-optimal algorithm for unlabeled
  planning}\label{sec:algorithm}
In this section we describe our algorithm for the unlabeled
multi-robot motion planning problem of unit-disc robots moving amid
polygonal obstacle and establish its completeness. Additionally, we
bound the cost of the returned solution. Finally, we analyze the
complexity of the algorithm.

\subsection{The algorithm}
We describe our recursive algorithm, which returns a set of $m$ paths
$\Pi$. Recall that the input consists of $S,T$ and a workspace $\W$,
which induces the free space $\F$. The algorithm first produces the
optimal-assignment path set $\Gamma$.  Let \opt be the optimal
solution cost and note that $|\Gamma|\leq \text{\opt}$ since $|\Gamma|$ is a
lower bound on the actual cost, as interactions between the robots
might increase the traversed distance.

Let $t_k$ be a standalone goal, which exists according to
Theorem~\ref{thm:standalone}.  Suppose that the $0$-hop path
$\gamma_k$ is not blocked by any other robot located in a start
position. Then, $\gamma_k$ is added to $\Pi$ and the algorithm is run
recursively on the input
$S':=S\setminus \{s_k\}, T':=T\setminus \{t_k\}$, with the workspace
$\W':=\W\setminus \B_1(t_k)$, which results in the free space
$\F':=\F\setminus \B_2(t_k)$.

Alternatively, in case that $\gamma_k$ is blocked, i.e., in
interference with some $s_i\in S, i\neq k$, the algorithm produces the
$1$-hop path $H_i^k$, as described in Theorem~\ref{thm:1_hop}, and
adds it to $\Pi$. Then, the algorithm is run recursively on the input
$S':=S\setminus \{s_i\}, T':=T\setminus \{t_k\}$, with the workspace
$\W':=\W\setminus \B_1(t_k)$, and the free space
$\F':=\F\setminus \B_2(t_k)$.

\subsection{Completeness and near-optimality}
We first show that the algorithm is guaranteed to find a solution, if
one exists, or report that none exists otherwise.

\begin{thm}\label{thm:complete}
  Given an input $S,T,\W$, which complies with
  assumptions~\ref{equation:separation}
  ,\ref{equation:obstacle-clearance} (Section~\ref{sec:simply}), and
  for which the number of start and goal positions for every connected
  component of $\F$ is the same, the algorithm is guaranteed to find a
  solution for the unlabeled multi-robot motion-planning problem.
\end{thm}
\begin{proof}
  Consider the first iteration of the algorithm. Let
  $\Gamma:=\{\gamma_1,\ldots, \gamma_m\}$ be the optimal-assignment
  path set and let $t_k$ be a standalone goal.

  Suppose that $\gamma_k$ is a $0$-hop path. Then, for every
  $j\neq k$, $\gamma_j\subset \F\setminus \B_2(t_k)$, since $t_k$ is a
  standalone goal.  Now suppose that $\gamma_k$ is not a $0$-hop
  path. By Theorem~\ref{thm:1_hop} the $1$-hop path $H_i^k$ from $s_i$
  to $t_k$ does not collide with any other start position. By
  Lemma~\ref{lem:switch} $\tilde{H}_k^i\subset \F\setminus
  \B_2(t_k)$.
  Additionally, notice that for every $j\neq i,k$,
  $\gamma_j\subset \F\setminus \B_2(t_k)$.

  Thus, in any of the two situations, the removal of the standalone
  goal does not separate between start and goal configurations that
  are in the same connected component of $\F$. Note that in the first
  level of the recursion the existence of a standalone goal $t_k$
  guarantees the existence of a path to $t_k$ which does not collide
  with the other robots. This is possible due to
  assumptions~\ref{equation:separation},
  \ref{equation:obstacle-clearance}. Thus, in order to ensure the
  success of the following recursions, we need to show that these
  assumptions are not violated. First, note that
  assumption~\ref{equation:separation} is always maintained, since we
  do not move existing start and goal positions. Secondly, when a
  robot situated in $t_k$ is treated as an obstacle added to the set
  $\O:=\O\cup \B_1(t_k)$ (or conversely removed from $\F$, as
  described above), the first assumption induces the second.
\end{proof}

We proceed to prove the near-optimality of our solution.

\begin{thm}\label{thm:optimal}
  Let $\Pi$ be the solution returned by our algorithm and let $\Gamma$ be
  the optimal-assignment path set for $S,T,\W$. Then $|\Pi|\leq
  \text{\opt}+4 m$, where \opt is the optimal solution cost.
\end{thm}

\begin{proof}
  Let $\Gamma:=\{\gamma_1,\ldots,\gamma_m\}$ be the optimal-assignment
  path set for the input $S,T,\W$, and let $t_k$ be a
  standalone. Additionally, assume that the algorithm generated the
  $1$-hop path $H_i^k$, since $\gamma_k$ was blocked (the case when
  $\gamma_k$ is not blocked is simpler to analyze). Let
  $\tilde{H}_k^i$ be the switch path from $s_k$ to $t_i$.  Similarly,
  denote by $\Gamma'$ the optimal-assignment path set for
  $S':=S\setminus \{s_i\},T':=T\setminus \{t_k\}$ and the workspace
  $\W':=\W\setminus \B_1(t_k)$, which induces the free space $\F'$.

  We will show that $|\Gamma'|+|H_i^k|\leq \text{\opt}+4$. As mentioned
  in Theorem~\ref{thm:complete}, for every $j\neq i,k$, it follows
  that for $\gamma_j\in \Gamma$ and $\gamma_j\subset \F'$. We also
  showed that the same holds for $\tilde{H}_k^i$, i.e.,
  $\tilde{H}_k^i\subset \F'$. Thus, the path set
  $R:=(\Gamma\setminus\{\gamma_i,\gamma_k\}) \cup \{\tilde{H}_k^i\}$
  represents a valid assignment for $S',T',\F'$, even though it might
  not be optimal. This means that $|\Gamma'|\leq |R|$. Thus,
  \begin{align*}
    |\Gamma'|+|H_i^k|& \leq  |R|+|H_i^k|  =  \sum_{j\neq
                       i,k}^m|\gamma_j|+|\tilde{H}_k^i|+|H_i^k|\\ & \leq 
                                                            \sum_{k=1}^m|\gamma_k|+4
                                                            = |\Gamma|+4
                                                            \leq \text{\opt}+4,
  \end{align*}
  \noindent where the third step is due to
  Corollary~\ref{cor:cost_1_hop}. Thus, every level of the recursion
  introduces and additive error factor of $4$. Repeating this process
  for the $m$ iterations we obtain $|\Pi|\leq \text{\opt}+4m$.
\end{proof}

\subsection{Complexity Analysis}
We analyze the complexity of the algorithm. In order to do so, we have
to carefully consider the operations that are performed in every
iteration of the algorithm.

\begin{thm}\label{thm:complexity}
  Given $m$ unlabeled unit-disc robots operating in a polygonal
  workspace with $n$ vertices, the algorithm described above returns a
  solution, or reports that none exists otherwise, with a running time
  of \complexity.
\end{thm}

\begin{proof}
  Let us consider a specific iteration $j$. The input of this
  iteration consists of $m-j+1$ start positions $S_j$ and $m-j+1$ goal
  positions $T_j$. The workspace region of this iteration is defined
  to be $\W_j:=\W\setminus \left(\B_1(T\setminus T_j)\right)$, which
  induces the free space~$\F_j$. Note that
  $S_1\equiv S, T_1\equiv T, \F_1\equiv \F, \Gamma_1\equiv \Gamma$.

  In order to find the optimal-assignment path set $\Gamma_j$ for
  $S_j,T_j,\F_j$ one has to first find the shortest path in $\F_j$ for
  every pair of start and goal positions $s\in S_j, t\in T_j$. Given
  the costs of all those combinations, the Hungarian
  algorithm~\cite{l-combi76}\footnote{Also known as the Kuhn-Munkres
    algorithm.} finds the optimal assignment, and so $\Gamma_j$ is
  produced. In the next step, a standalone goal $t_k$ is identified
  and it is checked whether the $0$-hop path leading to $t_k$ is
  clear, in which case $\gamma_k$ is added to $\Pi$. If it is not, one
  needs to find the last start position that interferes with this path
  and generate the respective $1$-hop path, which will be included in
  $\Pi$.

  The complexity of finding a shortest path for a disc depends on the
  complexity of the workspace, which in our case is $O(n+j)$. This
  task is equivalent to finding a shortest path for a point robot
  in~$\F_j$. The generation of $\F_j$ can be done in
  $O\left((n+j)\log^2(n+j)\right)$ time complexity,
  and the overall complexity of this structure would be
  $O(m+j)$~\cite{klps-jordan86}. A common approach for finding
  shortest paths in the plane is to construct a \emph{visibility
    graph}~\cite{cg_book} which encapsulates information between every
  pair of vertices of a given arrangement of segments, while avoiding
  crossings with the segments. In our case, the arrangement should
  include $\F_j$ as well as all the points from $S_j$ and $T_j$. Thus,
  we would need to generate a visibility graph over a generalized polygon of
  total complexity $O(m+n)$. This can be done in
  $\tilde{O}\left((m+n)^2\right)
  =\tilde{O}(m^2+n^2)$~\cite{cg_book}\footnote{We
    note that in our setting the arrangement consists not only of
    straight-line segments, but also of circular arcs which are
    induced by robots situated in target positions. Yet, the algorithm
    constructing the visibility graph can treat such cases as well, while
    still guaranteeing the (near-)quadratic run-time complexity mentioned
    above.}. Given the visibility graph with $O(m+n)$ vertices and
  $O(m^2+n^2)$ edges we find for each $s\in S_j$ the shortest path to
  every $t\in T_j$. For each $s$ we run the Dijkstra algorithm which
  requires $O(m^2+n^2)$ time, and finds the shortest path from the
  given $s$ to any $t\in T_j$. Since $|S_j|=m-j+1$, the total running
  time of finding the shortest path from every start to every goal is
  $O((m-j)(m^2+n^2))$.

  To find $\Gamma_j$ we employ the Hungarian
  algorithm~\cite{l-combi76}, which runs in $O\left((m-j)^3\right)$
  time.  Now, given $\Gamma_j$ we wish to find a standalone goal
  $t_k$. We first note that the complexity of each path in $\Gamma_j$
  is bounded by the complexity of $\F_j$, which is $O(n+j)$. For every
  $t_i\in T_j, \gamma_{i'}\in \Gamma_j, i\neq i'$ we check whether
  $\B_1(t_i)\cap \B_1(\gamma_{i'})\neq \emptyset$. This step has a
  running time of
  $\tilde{O}((m-j)(m-j)(n+j))=\tilde{O}((m-j)^2(n+j))$.  Finding the
  last closest blocking start from $S_j$ of the path $\gamma_k$ takes
  additional $O((m-j)(n+j))$ time by going over all starts in $S_j$
  and comparing the distance of their interference point with
  $\gamma_k$. Thus, the overall complexity of a given iteration $j$ is
  \begin{dmath*}
    \tilde{O}\left((m-j)(m^2+n^2)+(m-j)^3+(m-j)^2(n+j)+(m-j)^2\right)\\=
    \tilde{O}\left((m-j)(m^2+n^2)\right).
  \end{dmath*}
  Note that the cost of the different components is absorbed in the
  cost of calculating the shortest paths.  We conclude with the
  running time of the entire algorithm:
  \begin{dmath*} \tilde{O}\left(\sum_{j=1}^m (m-j)(m^2+n^2) \right)=
    \tilde{O}\left(m^4+m^2n^2 \right).
  \end{dmath*}
\end{proof}

\begin{figure*}[!]\centering
  \begin{subfigure}[b]{0.22\textwidth}
    \includegraphics[width=\textwidth]{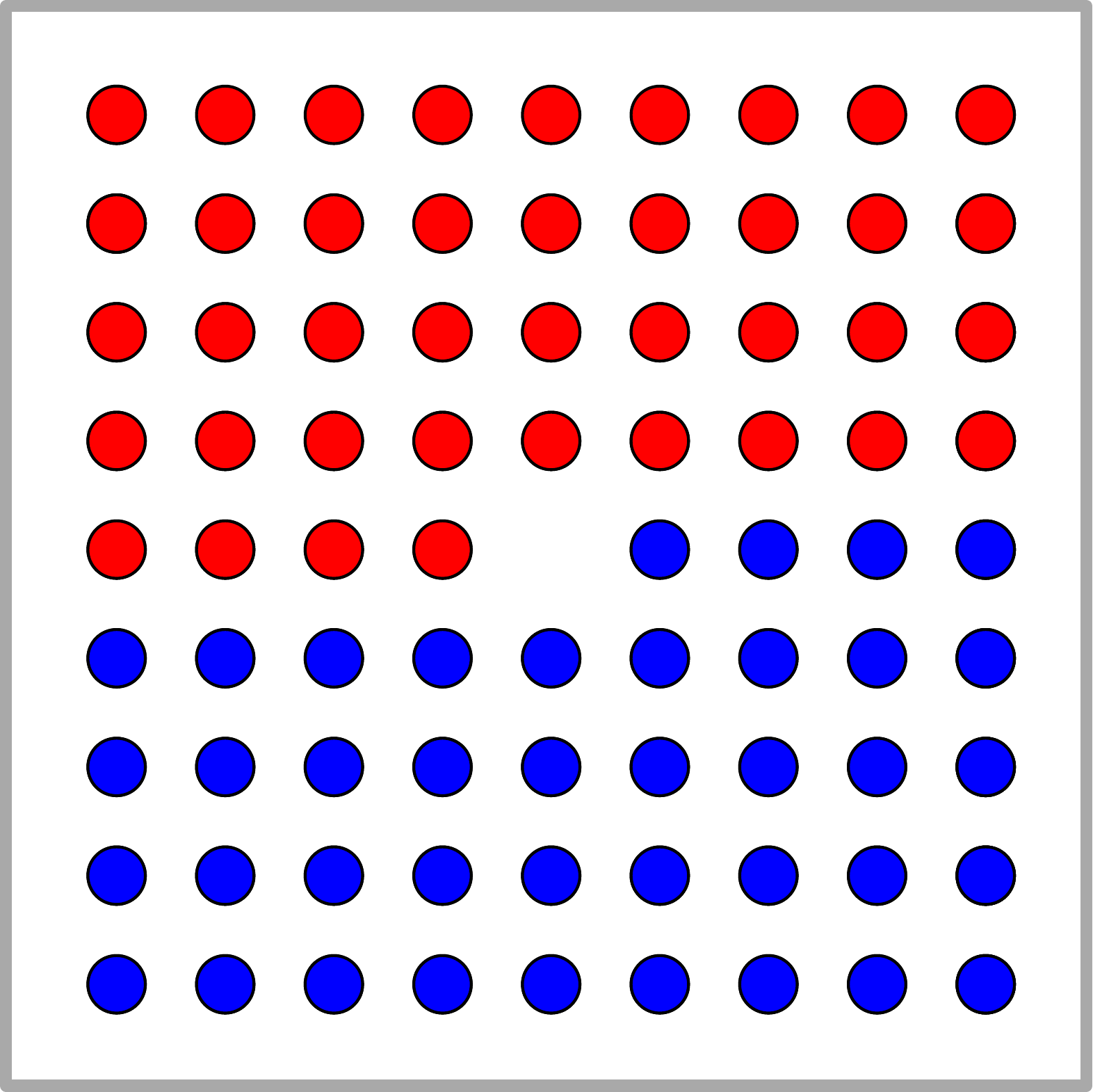}
    \caption{grid}
    \label{fig:scn_grid}
  \end{subfigure}\quad
  \begin{subfigure}[b]{0.22\textwidth}
    \includegraphics[width=\textwidth]{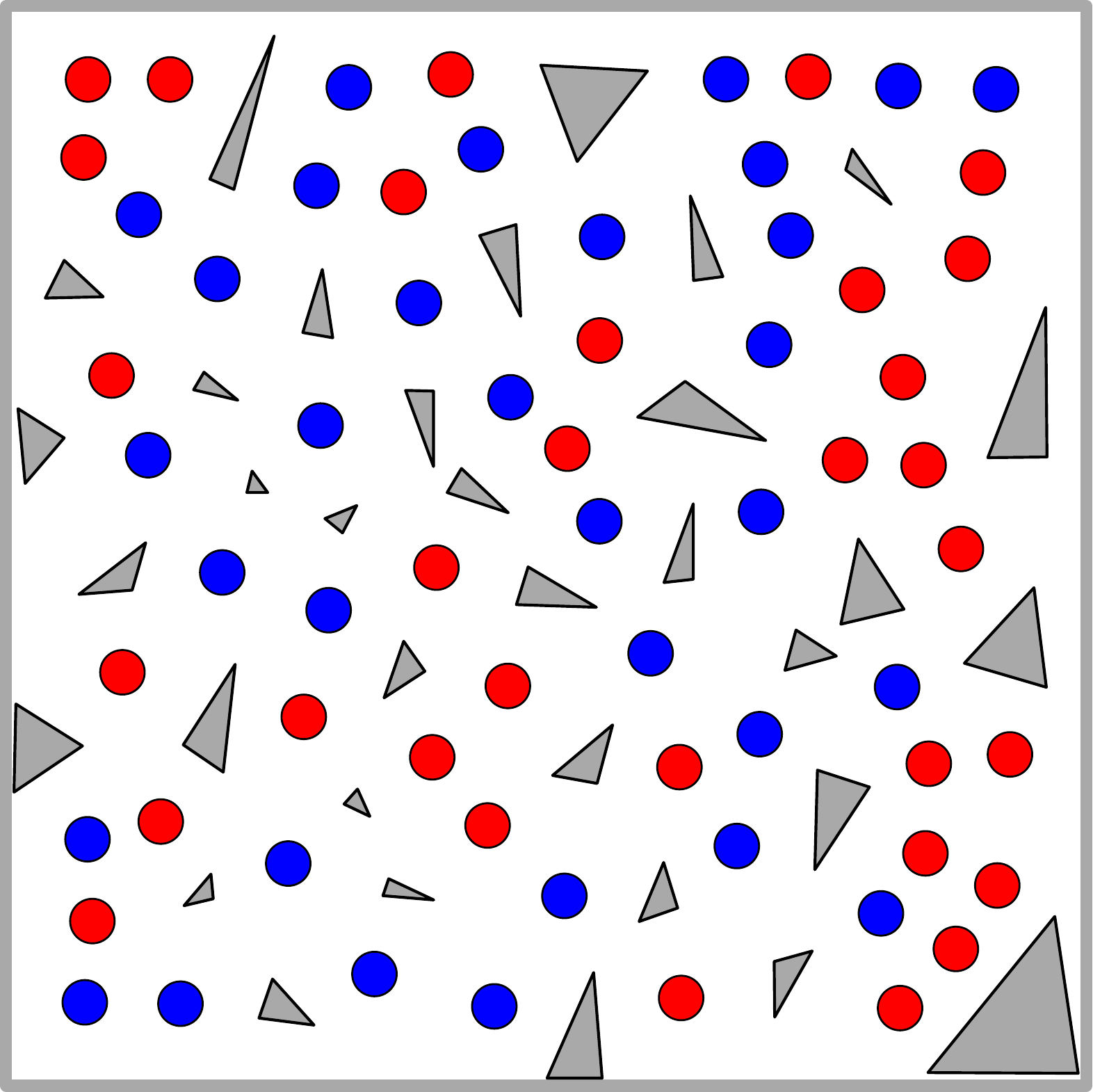}
    \caption{triangles}
    \label{fig:scn_triangles}
  \end{subfigure} \quad
  \begin{subfigure}[b]{0.22\textwidth}
    \includegraphics[width=\textwidth]{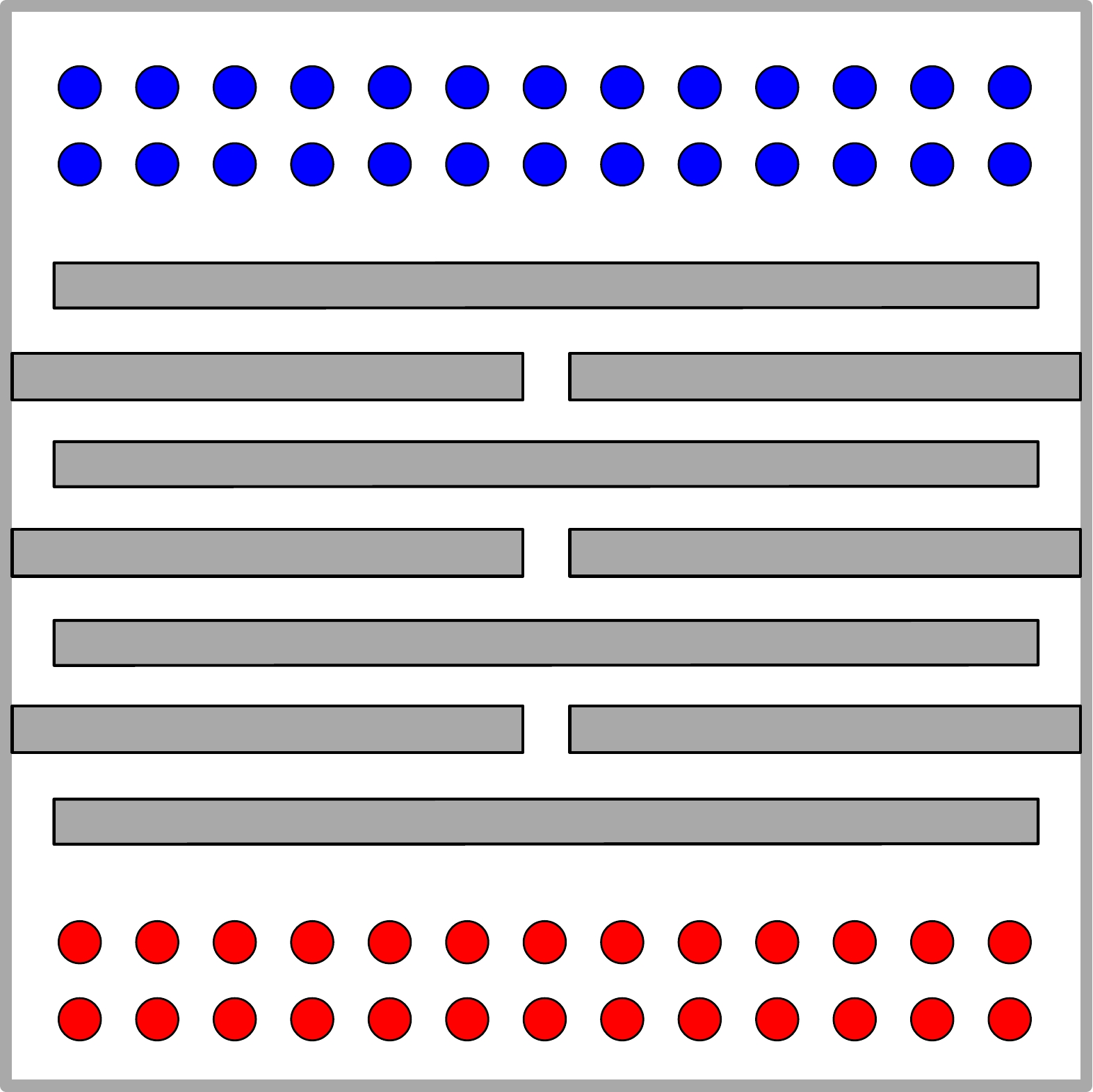}
    \caption{maze}
    \label{fig:scn_maze}
  \end{subfigure} \quad
  \begin{subfigure}[b]{0.22\textwidth}
    \includegraphics[width=\textwidth]{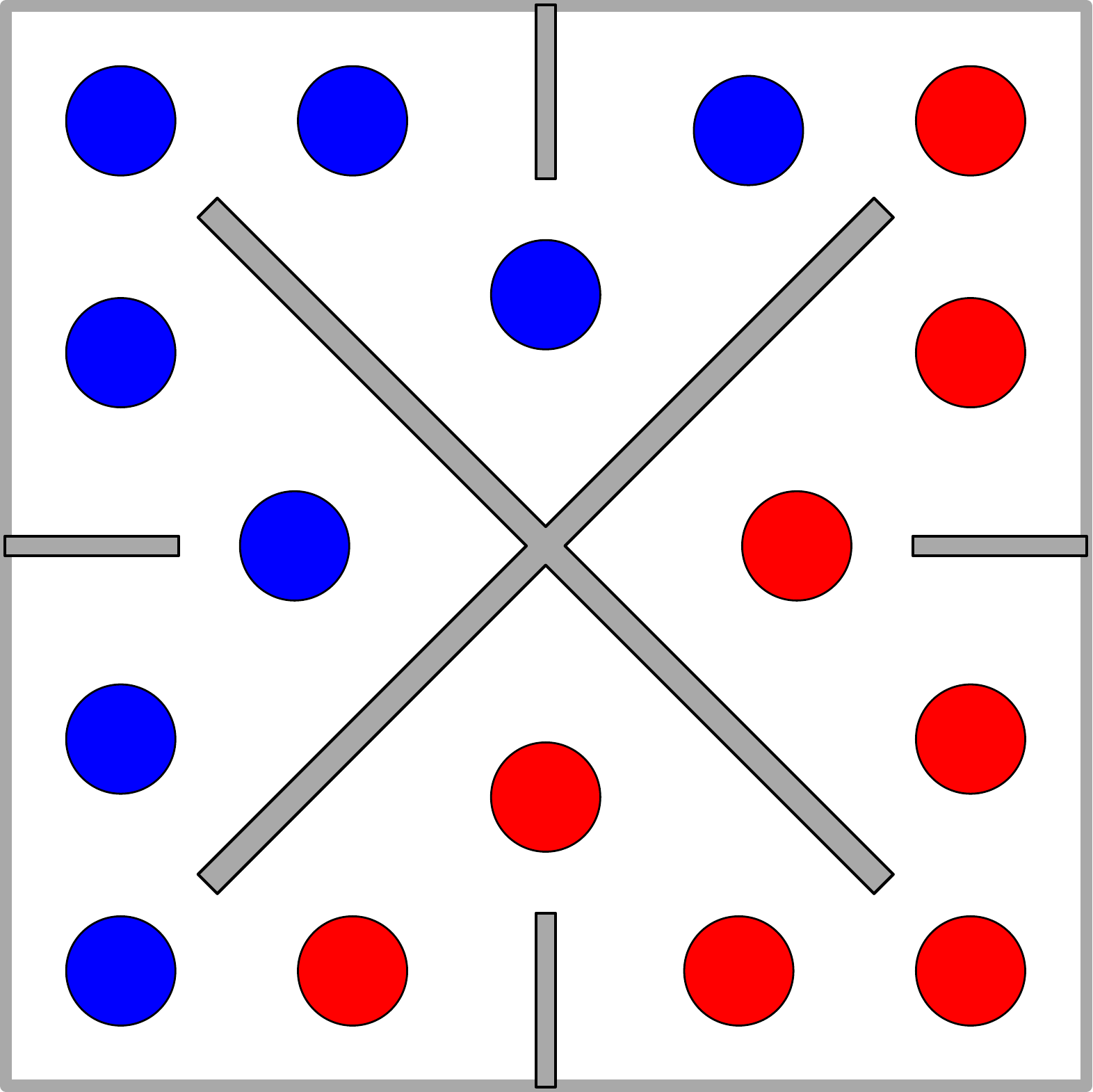}
    \caption{cross}
    \label{fig:scn_cross}
  \end{subfigure}
  \caption{Test scenarios. Obstacles are represented as gray
    polygons. Discs represent robots placed in start (red) and goal
    (blue) positions. (a) Grid scenario with 40 robots. (b) Triangles
    scenario with 32 robots and multiple triangular obstacles. (c)
    Maze scenario with with 26 robots; robots need to pass through a
    collection of narrow passages in order to reach their goals. (d)
    Cross scenario with 8 robots and several wall obstacles.}
  \label{fig:scenarios}
\end{figure*}

\section{Experimental results}\label{sec:experiments}
We implemented the algorithm and evaluated its performance on various
challenging scenarios.

\subsection{Implementation details}
First, we wish to emphasize that our implementation deals with
geometric primitives, e.g., polygons and discs, and does
not rely on any graph discretization of the problem. The
implementation relies on exact geometric methods that are provided by
CGAL~\cite{cgal}. As such, it is complete, robust and
deterministic. In addition to that, the implementation is
parameter-free.

We implemented the algorithm in \Cpp and relied heavily on CGAL for
geometric computing, and in particular on the {{\tt Arrangement\_2}
  package~\cite{cgal_arr_book}. Generation of the free space was done
  using \emph{Minkowski sums}, while shortest paths were generated
  using \emph{visibility graphs}~\cite{cg_book}. In order to find the
  optimal assignment, we used a \Cpp implementation of the \emph{Hungarian
    algorithm}~\cite{l-combi76}, available at~\cite{munkres}.

\subsection{Test scenarios}
We report in Table~\ref{tbl:results} on the running time of the
algorithm for the four scenarios depicted in
Figure~\ref{fig:scenarios}. The grid scenario
(Figure~\ref{fig:scn_grid}) is used to illustrate the performance of
the algorithm in a sterile obstacles-free workspace. The triangles and
cross scenarios (Figure~\ref{fig:scn_triangles},\ref{fig:scn_cross})
include multiple obstacles, which have a tremendous affect on the
performance of the algorithm. The maze scenario
(Figure~\ref{fig:scn_maze}), which also includes multiple obstacles,
has several narrow passages.

It is evident that the running time is dominated by shortest-path
calculations (see Theorem~\ref{thm:complexity}). The second largest
contributor to the running time of the algorithm, is the calculation
and maintenance of the configuration space, which includes the update
of the free space for every iteration.  A interesting relation is
found between the overall complexity of finding the optimal assignment
($O(m^4)$) and its modest contribution to the actual running time in
practice, when compared to the other components.  This can be
explained by the fact that the implementation of the Hungarian
algorithm~\cite{munkres} employs floating-point arithmetic, while the
geometric operations, e.g., shortest pathfinding, rely on exact
geometric kernels, which have unlimited
precision~\cite{MetMehPioSchYap08}. It is noteworthy that the
algorithm produces solutions whose cost is very close to the optimal
cost.

To gather a better understanding of the running time of the algorithm,
we also report on the running time of each iteration for the triangles
scenario (see Figure~\ref{fig:graph_iterations}). For every iteration
$1\leq j\leq 32$ we report on the running time of maintaining the
configuration space, and finding shortest paths.

Another important aspect of the algorithm is the relation between the
number of robots and the performance. To test this relation we
performed the following experiment, which is based on the triangles scenario
as well. We report in Figure~\ref{fig:graph_robots} how the performance
is affected by an increase or decrease in the number of robots. In
order to maintain a similar level of density between the various tests
we increase the radius of the robots to the maximal allowed radius
which abides by the separation assumptions.

The algorithm and its current implementation can deal with rather
complex scenarios. However, it is clear that there is a limit to the
number of robots or workspace complexity with which it can cope, due
the relatively high degree of the polynomial, in the running time
complexity.

\begin{figure*} \centering
  \begin{subfigure}[b]{0.43\textwidth}
    \includegraphics[width=\textwidth,clip=true, trim=10pt 10pt 40pt 19pt]{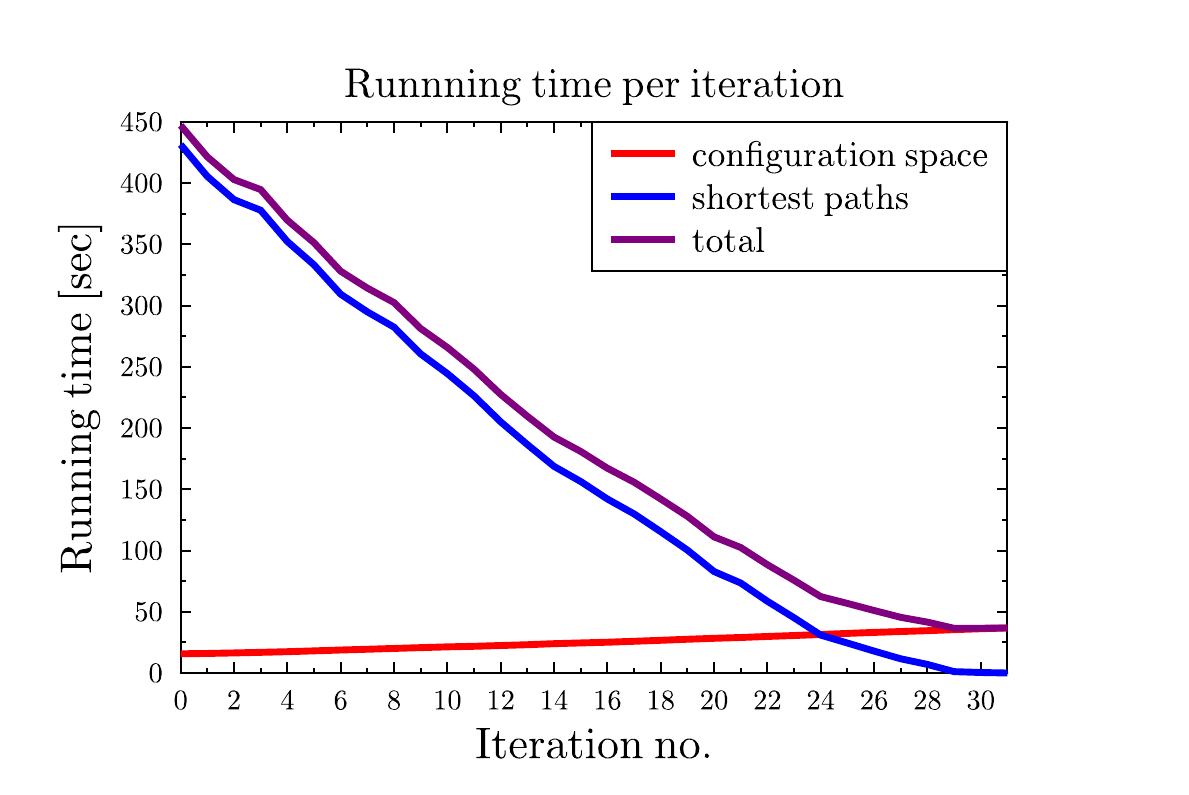}
    \caption{ }
    \label{fig:graph_iterations}
  \end{subfigure}\qquad 
  \begin{subfigure}[b]{0.43\textwidth}
    \includegraphics[width=\textwidth,clip=true, trim=10pt 10pt 40pt 19pt]{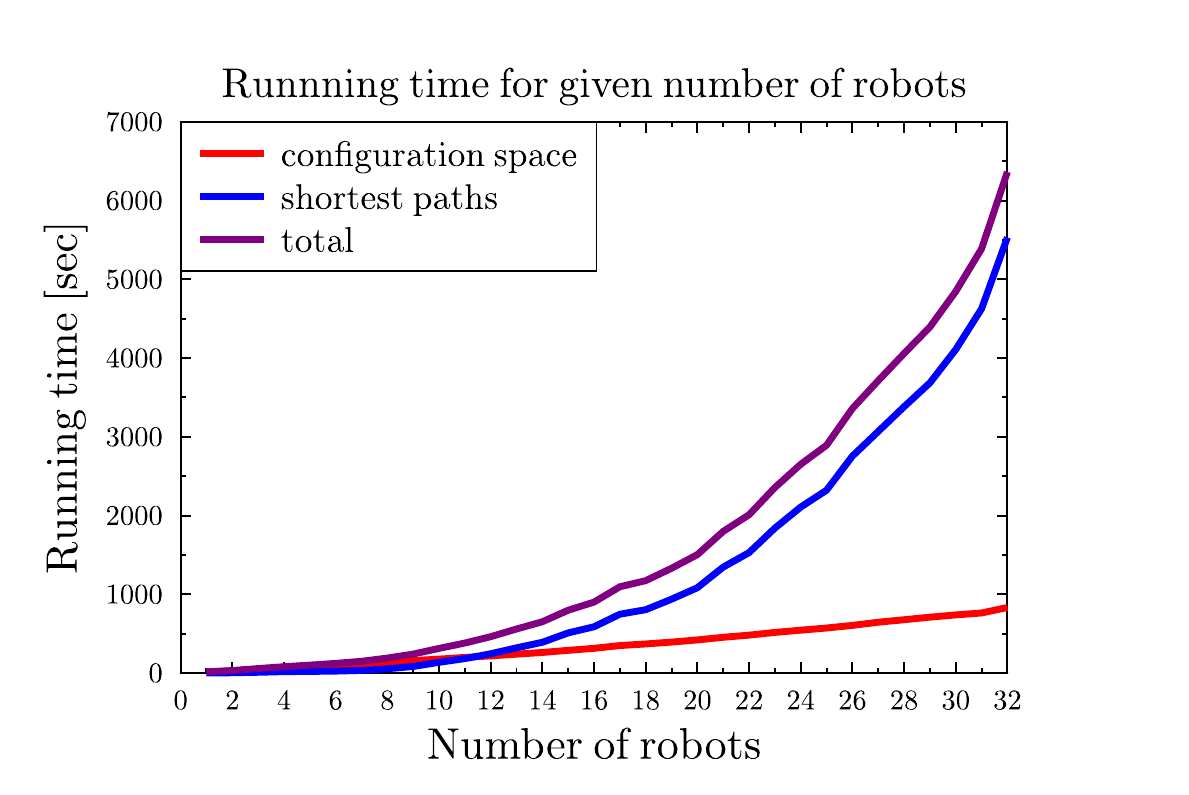}
    \caption{ }
    \label{fig:graph_robots}
  \end{subfigure}
  \caption{Graphs depicting the running time of maintaining the
    configuration space, and finding shortest path, for the triangles
    scenario (Figure~\ref{fig:scn_triangles}). (a) Running time is
    reported separately for every iteration of the algorithm. (b) The
    behavior of the algorithm for a growing number of robots in the
    triangles scenario is depicted.}\label{fig:graphs}
\end{figure*}

\begin{table*}[t]
\begin{center}
	\begin{tabular}{cc|c|c|c|c|}
      \cline{3-6}
      &  & \multicolumn{4}{c|}{Scenarios} \\ 
      \cline{3-6}
      &  & Grid  & Triangles & Maze & Cross \\ 
      \hhline{======}
      \multicolumn{1}{|c|}{\multirow{3}{*}{scenario properties} } & robots & 40 & 32 & 26 & 8 \\ 
      \cline{2-6}
      \multicolumn{1}{|c|}{} & workspace vertices & 4 & 118 & 56 & 44
      \\ 
      \cline{2-6}
      \multicolumn{1}{|c|}{} & robot radius & 0.05 & 0.034 & 0.035 & 0.09 \\
      \hhline{======} 
      \multicolumn{1}{|c|}{\multirow{5}{*}{running time (sec.)} } &
                                                                    configuration
                                                                    space
         &  261 & 832 & 141 & 17 \\ 
      \cline{2-6} 
      \multicolumn{1}{|c|}{} & shortest paths &  50  & 5532  & 658 & 1 \\ 
      \cline{2-6} 
      \multicolumn{1}{|c|}{} & optimal assignment  & 0.016 & 0.001 & 0.016 & 0 \\ 
      \cline{2-6} 
      \multicolumn{1}{|c|}{} & standalone goal  & 0.001 & 0.016 &
                                                                  0.119 & 0 \\ 
      \cline{2-6} 
      \multicolumn{1}{|c|}{} & \textbf{total} & \textbf{311} & \textbf{6394} & \textbf{800} & \textbf{19} \\ 
      \hhline{======}
      \multicolumn{1}{|c|}{\multirow{2}{*}{cost} } & lower bound &
                                                                   36.33
      & 10.41 & 154.86 & 12.11  \\ 
      \cline{2-6} 
      \multicolumn{1}{|c|}{} & actual cost  & 37.16 & 10.69 & 155.21 & 12.4 \\ 
      \hhline{======}
      \multicolumn{1}{|c|}{\multirow{2}{*}{algorithm's behavior} } &
                                                                     $0$-hop paths  & 25 & 32 & 14 & 7 \\ 
      \cline{2-6} 
      \multicolumn{1}{|c|}{} & $1$-hop paths  & 1 & 0 & 12 & 1 \\ 
      \hline 
\end{tabular} 
\end{center}
\caption{Results of our algorithm for the scenarios depicted in
  Figure~\ref{fig:scenarios}. We first describe the properties of each
  scenario, which consist of the number of robots, the complexity of
  the workspace ($n$), and the robot's radius (every
  scenario is bounded by the $[-1,1]^2$ square). Then, we report the running time
  (in seconds)
  of the different components of the algorithm: construction and
  maintenance of the configuration space; calculation of shortest
  paths; calculation of the optimal assignment; search for a
  standalone goal. Then, we report the lower-bound cost of the
  solution and the cost of the actual solution returned by the algorithm. Finally, we report the number of iterations for which
  the algorithm returned a $0$-hop or a $1$-hop path.}
\label{tbl:results}
\end{table*}

\section{Discussion and future work}\label{sec:discussion}
\subsection{Relaxing the separation conditions}
The algorithm presented in this paper requires that two conditions
will hold. First it is assumed that every pair of start or goal
position will be separated by a distance of at least $4$. The second
condition requires that every start or goal positions will be
separated from an obstacle by a distance of at least
$\sqrt{5}\approx 2.236$.  We believe that the obstacles-separation
factor can be lowered to $\sqrt{13-6\sqrt{3}}\approx 1.614$ using a
tighter mathematical analysis. While the proof for
Lemma~\ref{lem:shortcuts} would change, the rest of the proofs will
not require any special modification.  In the near future we aim to
consider less strict separation requirements such as reducing the
robots separation to $3$ and removing the requirement that start
positions will be separated from goal positions. However, it seems
that the machinery described here will not suffice for the tighter
setting and different tools will be required.

We note that without adequate obstacle clearance, it is in fact
impossible for any algorithm to guarantee a maximum $O(m)$ deviation
from the cost of an optimal-assignment path set, regardless of the
separation between start and goals. For instance, see the gadget in
Figure~\ref{figure:obstacle-clearance}. The two pairs of (blue)
start and (red) goal positions satisfy the $4$ separation
requirement. The disc starting at the bottom left, which we call disc
$1$, is matched by the optimal assignment with the goal on the top
left of the figure, and the disc starting in the middle, which we call
disc $2$, is matched with the goal on the right. This is due to the
fact that disc $1$ and disc $2$ reside in two disjoint components of
the free space. For disc $1$ to go through the lower tunnel, disc 2
must move non-trivially ( e.g., to the location of the green
discs). Let this nontrivial distance be $\delta$. Then, as disc $1$
passes through the upper tunnel, another nontrivial move of distance
$\delta$ must be performed by disc $2$ to clear the way. Now, we
remove disc $1$ from the gadget and stack $m/2$ of the resulting
gadget vertically. If we require the rest $m/2$ discs to
pass through the stacked construction, a total distance penalty of
$\Omega(m^2\delta)$ is incurred.

\begin{figure}[htp]
  \begin{center}
    \includegraphics[width=3in]{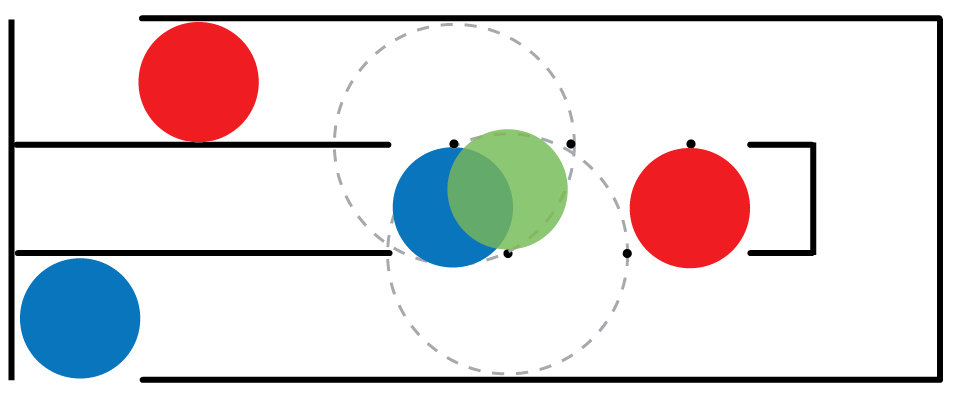}
  \end{center}
  \caption{Illustration of a scenario in which every solution has a
    deviation of $\Omega(m^2)$ from the cost of the optimal-assignment
    path set.}\label{figure:obstacle-clearance}
\end{figure}

\subsection{Improving the algorithm}
It is evident from our experiments that most of the running time of
the algorithm in practice is devoted to finding shortest paths. To be
specific, in each iteration the algorithm performs a shortest-path
search between every pair of start and target positions. This entire
process is performed only for the sake of finding a single standalone
goal. It will be interesting to investigate whether more information
can be extracted from the optimal assignment returned by the Hungarian
algorithm. For instance, could there be two or more standalone goals
to which there are separate non-interfering paths?

\subsection{Implementation in 3D}
For simplicity of presentation, our algorithm is discussed for the
planar setting. Yet, its completeness and near-optimality guarantees
hold also for the three-dimensional case with ball robots under the
two separation assumptions used in this paper. A main requirement for
a correct implementation of the algorithm is the ability to produce
shortest paths for a single robot moving amid static obstacles. While
this task does not pose particular challenges in the planar case, it
becomes extremely demanding in 3D. In particular, finding a shortest
path for a polyhedral robot translating amid polyhedral obstacles is
known to be NP-Hard~\cite{CanRei87}. It would be interesting to
investigate whether our algorithm can be extended to work with
approximate shortest paths.

\bibliographystyle{IEEEtran}
\bibliography{bibilography.bib}

\begin{thebibliography}{10}
\providecommand{\url}[1]{#1}
\csname url@samestyle\endcsname
\providecommand{\newblock}{\relax}
\providecommand{\bibinfo}[2]{#2}
\providecommand{\BIBentrySTDinterwordspacing}{\spaceskip=0pt\relax}
\providecommand{\BIBentryALTinterwordstretchfactor}{4}
\providecommand{\BIBentryALTinterwordspacing}{\spaceskip=\fontdimen2\font plus
\BIBentryALTinterwordstretchfactor\fontdimen3\font minus
  \fontdimen4\font\relax}
\providecommand{\BIBforeignlanguage}[2]{{%
\expandafter\ifx\csname l@#1\endcsname\relax
\typeout{** WARNING: IEEEtran.bst: No hyphenation pattern has been}%
\typeout{** loaded for the language `#1'. Using the pattern for}%
\typeout{** the default language instead.}%
\else
\language=\csname l@#1\endcsname
\fi
#2}}
\providecommand{\BIBdecl}{\relax}
\BIBdecl

\bibitem{calinescu2008reconfigurations}
G.~Calinescu, A.~Dumitrescu, and J.~Pach, ``Reconfigurations in graphs and
  grids,'' \emph{SIAM Journal on Discrete Mathematics}, vol.~22, no.~1, pp.
  124--138, 2008.

\bibitem{YuLav12WAFR}
J.~Yu and S.~M. LaValle, ``Multi-agent path planning and network flow,'' in
  \emph{Workshop on the Algorithmic Foundations of Robotics ({WAFR}), MIT,
  Cambridge, Massachusetts, USA}, 2012, pp. 157--173.

\bibitem{TurMicKum13ICRA}
M.~Turpin, N.~Michael, and V.~Kumar, ``Concurrent assignment and planning of
  trajectories for large teams of interchangeable robots,'' in
  \emph{International Conference on Robotics and Automation ({ICRA})}, 2013.

\bibitem{YuLav12CDC}
J.~Yu and S.~M. LaValle, ``Distance optimal formation control on graphs with a
  tight convergence time guarantee,'' in \emph{Proceedings IEEE Conference on
  Decision \& Control}, 2012, pp. 4023--4028.

\bibitem{SolHal15}
K.~Solovey and D.~Halperin, ``On the hardness of unlabeled multi-robot motion
  planning,'' in \emph{Robotics: Science and Systems (RSS)}, 2015, these
  proceedings.

\bibitem{abhs-unlabeled14}
A.~Adler, M.~de~Berg, D.~Halperin, and K.~Solovey, ``Efficient multi-robot
  motion planning for unlabeled discs in simple polygons,'' in \emph{Workshop
  on the Algorithmic Foundations of Robotics ({WAFR})}, 2014.

\bibitem{TurMicKum12}
M.~Turpin, N.~Michael, and V.~Kumar, ``Trajectory planning and assignment in
  multirobot systems,'' in \emph{Workshop on the Algorithmic Foundations of
  Robotics ({WAFR}), MIT, Cambridge, Massachusetts, USA}, 2012, pp. 175--190.

\bibitem{KarGerSta12}
I.~Karamouzas, R.~Geraerts, and A.~F. van~der Stappen, ``Space-time group
  motion planning,'' in \emph{Workshop on the Algorithmic Foundations of
  Robotics (WAFR), MIT, Cambridge, Massachusetts, USA}, 2012, pp. 227--243.

\bibitem{klb-ftpp13}
A.~Krontiris, R.~Luna, and K.~E. Bekris, ``From feasibility tests to path
  planners for multi-agent pathfinding,'' in \emph{Symposium on Combinatorial
  Search, ({SOCS}), Leavenworth, Washington, USA}, 2013.

\bibitem{TurMohMicKum13}
M.~Turpin, K.~Mohta, N.~Michael, and V.~Kumar, ``Goal assignment and trajectory
  planning for large teams of aerial robots,'' in \emph{Robotics: Science and
  Systems}, 2013.

\bibitem{YuLav13ICRA-A}
J.~Yu and S.~M. LaValle, ``Planning optimal paths for multiple robots on
  graphs,'' in \emph{International Conference on Robotics and Automation
  ({ICRA})}, 2013, pp. 3612--3617.

\bibitem{WurDanMou08}
P.~R. Wurman, R.~D'Andrea, and M.~Mountz, ``Coordinating hundreds of
  cooperative, autonomous vehicles in warehouses,'' \emph{AI Magazine},
  vol.~29, no.~1, pp. 9--19, 2008.

\bibitem{hss-cmpmio}
J.~E. Hopcroft, J.~T. Schwartz, and M.~Sharir, ``On the complexity of motion
  planning for multiple independent objects; {PSPACE}-hardness of the
  ``{W}arehouseman's problem'','' \emph{International Journal of Robotics
  Research}, vol.~3, no.~4, pp. 76--88, 1984.

\bibitem{sy-snp84}
P.~G. Spirakis and C.-K. Yap, ``Strong {NP}-hardness of moving many discs,''
  \emph{Information Processing Letters}, vol.~19, no.~1, pp. 55--59, 1984.

\bibitem{schwartz1983piano}
J.~T. Schwartz and M.~Sharir, ``On the piano movers' problem: {III}.
  coordinating the motion of several independent bodies: the special case of
  circular bodies moving amidst polygonal barriers,'' \emph{International
  Journal of Robotics Research}, vol.~2, no.~3, pp. 46--75, 1983.

\bibitem{ErdLoz86}
M.~A. Erdmann and T.~Lozano-P\'erez, ``On multiple moving objects,'' in
  \emph{International Conference on Robotics and Automation ({ICRA})}, 1986,
  pp. 1419--1424.

\bibitem{LavHut98b}
S.~M. LaValle and S.~A. Hutchinson, ``Optimal motion planning for multiple
  robots having independent goals,'' \emph{IEEE Transactions on Robotics \&
  Automation}, vol.~14, no.~6, pp. 912--925, 1998.

\bibitem{PenAke02}
J.~Peng and S.~Akella, ``Coordinating multiple robots with kinodynamic
  constraints along specified paths,'' in \emph{Algorithmic Foundations of
  Robotics {V}}, J.-D. Boissonat, J.~Burdick, K.~Goldberg, and S.~Hutchinson,
  Eds.\hskip 1em plus 0.5em minus 0.4em\relax Berlin: Springer-Verlag, 2002,
  pp. 221--237.

\bibitem{BerOve05}
J.~van~den Berg and M.~H. Overmars, ``Prioritized motion planning for multiple
  robots,'' in \emph{International Conference on Intelligent Robots and Systems
  (IROS)}, 2005, pp. 430 -- 435.

\bibitem{GhrOkaLav05}
R.~Ghrist, J.~M. O'Kane, and S.~M. LaValle, ``{Computing Pareto Optimal
  Coordinations on Roadmaps},'' \emph{International Journal of Robotics
  Research}, vol.~24, no.~11, pp. 997--1010, 2005.

\bibitem{BerSnoLinMan09}
J.~van~den Berg, J.~Snoeyink, M.~Lin, and D.~Manocha, ``Centralized path
  planning for multiple robots: Optimal decoupling into sequential plans,'' in
  \emph{Robotics: Science and Systems}, 2009.

\bibitem{OdoLoz89}
P.~A. O'Donnell and T.~Lozano-P\'{e}rez, ``Deadlock-free and collision-free
  coordination of two robot manipulators,'' in \emph{International Conference
  on Robotics and Automation ({ICRA})}, 1989, pp. 484--489.

\bibitem{SveOve98}
P.~\v{S}vestka and M.~H. Overmars, ``Coordinated path planning for multiple
  robots,'' \emph{Robotics and Autonomous Systems}, vol.~23, no.~3, pp.
  125--152, 1998.

\bibitem{avbsv-mpfmr}
B.~Aronov, M.~de~Berg, A.~F. van~der Stappen, P.~\v{S}vestka, and J.~Vleugels,
  ``Motion planning for multiple robots,'' \emph{Discrete {\&} Computational
  Geometry}, vol.~22, no.~4, pp. 505--525, 1999.

\bibitem{KloHut06}
S.~Kloder and S.~Hutchinson, ``Path planning for permutation-invariant
  multirobot formations,'' \emph{IEEE Transactions on Robotics}, vol.~22,
  no.~4, pp. 650--665, 2006.

\bibitem{shh12}
O.~Salzman, M.~Hemmer, and D.~Halperin, ``On the power of manifold samples in
  exploring configuration spaces and the dimensionality of narrow passages,''
  in \emph{Workshop on the Algorithmic Foundations of Robotics (WAFR)}, 2012,
  pp. 313--329.

\bibitem{ssh-fne13}
K.~Solovey, O.~Salzman, and D.~Halperin, ``Finding a needle in an exponential
  haystack: Discrete {RRT} for exploration of implicit roadmaps in multi-robot
  motion planning,'' in \emph{Workshop on the Algorithmic Foundations of
  Robotics (WAFR)}, 2014.

\bibitem{WagCho15}
G.~Wagner and H.~Choset, ``Subdimensional expansion for multirobot path
  planning,'' \emph{Artif. Intell.}, vol. 219, pp. 1--24, 2015.

\bibitem{GriAke05}
E.~J. Griffith and S.~Akella, ``Coordinating multiple droplets in planar array
  digital microfluidic systems,'' \emph{International Journal of Robotics
  Research}, vol.~24, no.~11, pp. 933--949, 2005.

\bibitem{kms-cpmg}
D.~Kornhauser, G.~Miller, and P.~Spirakis, ``Coordinating pebble motion on
  graphs, the diameter of permutation groups, and applications,'' in
  \emph{Foundations of Computer Science (FOCS)}.\hskip 1em plus 0.5em minus
  0.4em\relax IEEE Computer Society, 1984, pp. 241--250.

\bibitem{ampp-ltafpm}
V.~Auletta, A.~Monti, M.~Parente, and P.~Persiano, ``A linear time algorithm
  for the feasibility of pebble motion on trees,'' in \emph{Scandinavian
  Symposium and Workshops on Algorithm Theory (SWAT)}, 1996, pp. 259--270.

\bibitem{gh-mcpm}
G.~Goraly and R.~Hassin, ``Multi-color pebble motion on graphs,''
  \emph{Algorithmica}, vol.~58, no.~3, pp. 610--636, 2010.

\bibitem{Yu13}
J.~Yu, ``A linear time algorithm for the feasibility of pebble motion on
  graphs,'' \emph{CoRR}, vol. abs/1301.2342, 2013.

\bibitem{k-cpmg}
D.~Kornhauser, ``Coordinating pebble motion on graphs, the diameter of
  permutation groups, and applications,'' M.{Sc}. Thesis, Department of
  Electrical Engineering and Computer Scienec, Massachusetts Institute of
  Technology, 1984.

\bibitem{LunBek11}
R.~Luna and K.~E. Bekris, ``An efficient and complete approach for cooperative
  path-finding,'' in \emph{Conference on Artificial Intelligence, San
  Francisco, California, USA}, 2011.

\bibitem{KatYuLav13}
M.~Katsev, J.~Yu, and S.~M. LaValle, ``Efficient formation path planning on
  large graphs,'' in \emph{2013 {IEEE} International Conference on Robotics and
  Automation, Karlsruhe, Germany}, 2013, pp. 3606--3611.

\bibitem{sh-kcolor14}
K.~Solovey and D.~Halperin, ``{\it k}-{C}olor multi-robot motion planning,''
  \emph{International Journal of Robotic Research}, vol.~33, no.~1, pp. 82--97,
  2014.

\bibitem{l-combi76}
E.~L. Lawler, \emph{Combinatorial optimization: networks and matroids}.\hskip
  1em plus 0.5em minus 0.4em\relax Courier Dover Publications, 1976.

\bibitem{klps-jordan86}
K.~Kedem, R.~Livne, J.~Pach, and M.~Sharir, ``On the union of jordan regions
  and collision-free translational motion amidst polygonal obstacles,''
  \emph{Discrete {\&} Computational Geometry}, vol.~1, pp. 59--70, 1986.

\bibitem{cg_book}
M.~de~Berg, O.~Cheong, M.~van Kreveld, and M.~Overmars, \emph{Computational
  Geometry: Algorithms and Applications}, 3rd~ed.\hskip 1em plus 0.5em minus
  0.4em\relax Springer-Verlag, 2008.

\bibitem{cgal}
``\textsc{Cgal}, {C}omputational {G}eometry {A}lgorithms {L}ibrary,''
  \url{http://www.cgal.org}.

\bibitem{cgal_arr_book}
E.~Fogel, D.~Halperin, and R.~Wein, \emph{{CGAL} Arrangements and Their
  Applications - {A} Step-by-Step Guide}, ser. Geometry and computing.\hskip
  1em plus 0.5em minus 0.4em\relax Springer, 2012, vol.~7.

\bibitem{munkres}
``Implementation of the {Kuhn-Munkres} algorithm,''
  \url{https://github.com/saebyn/munkres-cpp}.

\bibitem{MetMehPioSchYap08}
L.~Kettner, K.~Mehlhorn, S.~Pion, S.~Schirra, and C.~Yap, ``Classroom examples
  of robustness problems in geometric computations,'' \emph{Comput. Geom.},
  vol.~40, no.~1, pp. 61--78, 2008.

\bibitem{CanRei87}
J.~F. Canny and J.~H. Reif, ``New lower bound techniques for robot motion
  planning problems,'' in \emph{28th Annual Symposium on Foundations of
  Computer Science, Los Angeles, California, USA}, 1987, pp. 49--60.

\end{thebibliography}

\end{document}